\documentclass[10pt,conference]{IEEEtran}
\usepackage[british]{babel}
\usepackage{url}
\usepackage{dsfont}
\usepackage{amsmath,amssymb}
\usepackage{float}
\usepackage{algorithm}
\usepackage{graphicx}
\usepackage{framed}
\usepackage{xspace}
\usepackage{array}
\usepackage{pifont}
\newcommand{\cmark}{\ding{51}}%
\newcommand{\xmark}{\ding{55}}%
\usepackage[noend]{algpseudocode}
\usepackage{amsthm}
\usepackage{amssymb}
\usepackage{graphicx}
\usepackage{caption}
\usepackage{subfig}
\graphicspath {{figures/}}
\newtheorem{theorem}{Theorem}
\newtheorem{lemma}{Lemma}

\newcommand*{\EXTENDED}{}

\title{Optimal Elephant Flow Detection}
\author{\IEEEauthorblockN{Ran Ben Basat}
	\IEEEauthorblockA{Computer Science Department\\
		Technion\\
		sran@cs.technion.ac.il
	}
	\and
	\IEEEauthorblockN{Gil Einziger}
	\IEEEauthorblockA{Nokia Bell Labs\\
		gil.einziger@nokia.com
	}
	\and
	\IEEEauthorblockN{Roy Friedman }
	\IEEEauthorblockA{Computer Science Department\\
		Technion\\
		roy@cs.technion.ac.il
	}
	\and
	\IEEEauthorblockN{Yaron Kassner}
	\IEEEauthorblockA{Computer Science Department\\
		Technion\\
		kassnery@cs.technion.ac.il
	}
}

\begin{document}

\maketitle
\newcommand{\sizeEst}{{\sc {$\epsilon$ - Volume Estimation}}}
\newcommand{\heavyHitters}{{\sc {$(\theta, \epsilon)$ - Elephant Flows}}}
\newcommand{\matrixCellWidth}{0.33\linewidth} 

\begin{abstract}
Monitoring the traffic volumes of elephant flows, including the total byte count per flow, is a fundamental capability for online network measurements.
We present an asymptotically optimal algorithm for solving this problem in terms of both space and time complexity.
This improves on previous approaches, which can only count the number of packets in constant time.
We evaluate our work on real packet traces, demonstrating an up to X2.5 speedup compared to the best~alternative.
\end{abstract}
\newcolumntype{L}[1]{>{\raggedright\let\newline\\\arraybackslash\hspace{0pt}}m{#1}}
\newcolumntype{C}[1]{>{\centering\let\newline\\\arraybackslash\hspace{0pt}}m{#1}}
\setlength\tabcolsep{1pt}
{\renewcommand{\arraystretch}{1.4}
	\begin{table*}[t!]
		\centering
		\begin{tabular}{|C{3.5cm}|C{3cm}|C{3cm}|C{3cm}|c|C{3cm}|}
			\hline 
			Algorithm & Space & Query Time & Update Time & Deterministic\tabularnewline
			\hline 
			\hline 
			Space Saving~\cite{SpaceSavingIsTheBest} & $O\left(\frac{1}{\epsilon}\right)$ & $O(1)$ & $O\left(\log\frac{1}{\epsilon}\right)$ &\cmark
			\tabularnewline
			\hline 
			
			CM Sketch~\cite{CMSketch} & $O\left(\frac{1}{\epsilon} \cdot \log \frac{1}{\delta}\right)$ & $O\left(\log \frac{1}{\delta}\right)$ & $O\left(\log \frac{1}{\delta}\right)$ & \xmark\tabularnewline
			\hline 
			Count Sketch~\cite{CountSketch} & $O\left(\frac{1}{\epsilon^2} \cdot \log \frac{1}{\delta}\right)$ & $O\left(\log \frac{1}{\delta}\right)$ & $O\left(\log \frac{1}{\delta}\right)$ & \xmark\tabularnewline
			\hline 
			This Paper- IM-SUM & $O\left(\frac{1}{\epsilon}\right)$ & $O(1)$ & $O(1)$ amortized & \cmark\tabularnewline
			\hline 
			This Paper- DIM-SUM & $O\left(\frac{1}{\epsilon}\right)$ & $O(1)$ & $O(1)$ worst case & \cmark\tabularnewline
			\hline 
		\end{tabular}\smallskip
		\caption{Comparison of frequency estimation algorithms that support weights. }
		\label{tbl:comparison}
	\end{table*} 
\section{Introduction}
\subsection{Background}

Network monitoring is at the core of many important networking protocols such as load balancing~\cite{DynamicFlow,LoadBalancing}, traffic engineering~\cite{TrafficEngeneering,TCPRetransmissions}, routing, fairness~\cite{ApproximateFairness}, intrusion and anomaly detection~\cite{IntrusionDetection2,IntrusionDetection,EXAnomalyDetection}, caching~\cite{TinyLFU}, policy enforcement~\cite{SLA} and performance diagnostics~\cite{DevFlow}.
Effective network monitoring requires maintaining various levels of traffic statistics, both as an aggregate and on a per-\emph{flow} basis.
This includes the number of distinct flows, also known as \emph{flow cardinality}, the number of packets generated by each flow, and the total traffic volume attributed to each flow.
Each of these adds complementing capabilities for managing and protecting the network.
For example, a large increase in cardinality may indicate a port-scanning attack, while statistics about the number of packets or the volume of traffic can help perform load balancing, meet QoS guarantees and detect denial-of-service (DoS) attacks.
For the latter, identifying the \emph{top-$K$} flows, or the \emph{heavy-hitters} and \emph{elephant flows} are essential competencies.
Similarly, traffic engineering~\cite{TrafficEngeneering} involves detecting high volume flows and ensuring that they are efficiently routed.

Operation speed is of particular importance in network measurement.
For example, a reduction in latency for partition/aggregate workloads can be achieved if we are able to identify traffic bursts in near real time~\cite{DevFlow}.
The main challenges in addressing the above mentioned tasks come from the high line rates and large scale of modern networks.
Specifically, to keep up with ever growing line rates, update operations need to be extremely fast.
As already mentioned, when near real-time decisions are expected, queries should also be answered quickly.
In addition, due to the huge number of flows passing through a single network device, memory is becoming a major concern.
In a hardware implementation, the data structures should fit in TCAM or SRAM, because DRAM is too slow to keep up with line rate updates.
Similarly, in a software implementation, as can be envisioned in upcoming SDN and NFV realizations, the ability to perform computations in a timely manner greatly depends on whether the data structures fit in the hardware cache and whether the relevant memory pages can be pinned to avoid swapping.

While the problem of identifying the top-$K$ flows and heavy-hitters in terms of the number of packets has been addressed by many previous works, e.g.,~\cite{WCSS,frequent4,BatchDecrement,SpaceSavings}, detection of elephant flows in terms of their traffic volume has received far less attention.
However, it is non-trivial to translate a top-$K$ or heavy-hitters algorithm to an efficient elephant flow solution,
because many of the former maintain ordered or semi-ordered data structures~\cite{WCSS,frequent4,BatchDecrement,SpaceSavings}.
Performing fast updates to these data structures depends on the fact that each packet increments (or decrements) its corresponding counter(s) by~$1$.
Hence, the perturbation caused by each update is fairly contained, predictable and tractable.
In contrast, when each packet modifies its corresponding counter(s) by its entire size, such maintenance becomes much harder.
Further, treating an increment (or decrement) by the packet's size $S$ as a sequence of $S$ increments (or decrements) by $1$ would multiply the update time by a factor of $S$, rendering it too slow.

\subsubsection{Contributions}

In this paper, we introduce the first elephant flow detection scheme that provides the following benefits:
($i$) constant time updates, ($ii$) constant time point-queries, ($iii$) detection of elephant flows in linear time, which is optimal, and ($iv$) asymptotically optimal space complexity.
This is achieved with two hash-tables whose size is proportional to the number of elephant flows, and a single floating point variable.
We present two flavors of the algorithm for maintaining these data structures.
The first version \emph{Iterative Median SUMming (IM-SUM)}, which is easier to describe and simpler to code, works in amortized $O(1)$ time.
The second variant \emph{De-amortized Iterative Median SUMming (DIM-SUM)} de-amortizes the first, thereby obtaining worst case $O(1)$ execution time.

We also evaluate the performance of these variants on both synthetic and real-world traces, and compare their execution time to other leading alternatives.
We demonstrate that IM-SUM is up to 2.5 times faster than any of our competitors on these traces, and up to an order of magnitude faster than others.
DIM-SUM is slower than IM-SUM, but still faster than all competitors for small values of $\varepsilon$.
DIM-SUM is faster in terms of worst case guarantee, ensuring constant update time.
The latter is important when real-time behavior is required.

\subsection{Related work}
Network monitoring capabilities are needed in both hardware and software~\cite{SDN1}.
In hardware, space is a critical constraint as there is no sufficient memory technology; while SRAM is fast enough to operate at line speed, it is too small to accommodate all flows.
On the contrary, DRAM is too slow to be read at line speed.
Traditional network monitoring approaches utilized short probabilistic counters~\cite{ICEBuckets,CEDAR,CASE} to reduce the memory requirements.

\emph{Counter arrays} are managed by network devices~\cite{ICEBuckets,CEDAR,CASE} as well as \emph{sketches} such as the Count Sketch~\cite{CountSketch} and the Count Min Sketch (CM-Sketch)~\cite{CMSketch}. These algorithms can be extended to support finding frequent items using hierarchy and group testing~\cite{SpaceSavingIsTheBest,CMSketch,GroupTesting}.
As network line rates became faster, sketches evolved into more complex algorithms that require significantly less memory at the expense of a long decoding time.
Such algorithms include Counter Braids~\cite{CounterBraids}, Randomized Counter Sharing~\cite{RCS} and Counter Tree~\cite{CounterTree}.
While these algorithms handle updates very fast, they incur a long query time.
Thus, queries can only be done off-line and not in real time as required by some networking applications.

For software implementations, counter based algorithms are often the way to go~\cite{SpaceSavingIsTheBest}.
These methods typically maintain a flow table, where each monitored flow receives a table entry.
Counter algorithms differ from one another in the the flow table maintenance.
Specifically, in \emph{Lossy Counting}~\cite{LC}, new flows are always added to the table.
In order to keep the table size bounded, flow counters are periodically decremented and flows whose counters reach 0 are deleted.
Lossy counting is simple and effective, but its space consumption is not optimal.
Lossy Counting's space consumption was empirically improved with probabilistic eviction of entries~\cite{PLC} as well as statistical knowledge about the stream distribution~\cite{MLC} that allows dropping excess table entries~earlier.

\emph{Frequent (FR)}~\cite{BatchDecrement} is a space optimal algorithm~\cite{Kranakis03boundsfor}.  In FR, instead of decrementing counters periodically, when a packet arrives for a non resident flow and the table is full, all counters are decremented. This improves the space complexity to optimal, but flow counters are needlessly decremented and therefore the algorithm is less accurate.
This is solved by the Space Saving algorithm~\cite{SpaceSavings}.
In Space Saving, when a packet that belongs to a flow that does not have a counter arrives, and the flow table is full, the algorithm evicts the entry whose counter is minimal.
The rest of the counters are untouched and their estimation is therefore more accurate.
For packet counting, since all counter updates are $+1$, Space Saving and FR can be implemented in $O(1)$~\cite{WCSS,SpaceSavings}.
This makes Space Saving and FR asymptotically optimal in both space and time.
The version of Space Saving we compare against in this paper uses a heap to manage its counters~\cite{SpaceSavingIsTheBest,SpaceSavingIsTheBest2010,SpaceSavingIsTheBest2009}; we denote it hereafter SSH.
In the general case, packets have different sizes so SSH requires a logarithmic runtime.
This is preferred over the original implementation with ordered linked lists~\cite{SpaceSavings}, which requires a constant time to count the number of packets but a linear time when considering packet sizes.

Various methods such as a sketches for tail items and a randomized counter admission policy can make (empirically) more efficient data structures~\cite{FSS,RAP}.
These improvements are orthogonal to the one presented in this paper.

Alternatively,~\cite{SamplingSDN} implemented elephant flow identification with a \emph{Sample and Hold} like technique~\cite{LC}.
That is, the controller receives sampled packets and periodically updates the monitored flows to reflect the current state of the heavy hitters. Unfortunately, sampling discards some of the information about packet sizes that can be used to reduce the error. In addition, since sampling does not take into account the packet size, large packets can be missed, resulting in a large error.

In general, $\Omega\left(\frac{1}{\epsilon}\right)$ space is required to approximately count with an additive error bounded by $\varepsilon\cdot R$, where $R$ is the total weight of all items in the stream~\cite{SpaceSavings}.
The optimal runtime is $O(1)$ for both update and query.
Related works which are capable of estimating flow volume (rather than packet counting) are summarized in Table~\ref{tbl:comparison}.
As listed, our algorithms are the first to offer both (asymptotically) optimal space consumption and optimal runtime.

\section{Model}
We consider a stream ($\mathcal S$) of tuples of the form, $(a_i,w_i)$, where $a_i$ denotes item's id and $w_i$ its (non-negative) weight.
At each step, a new tuple is added to the stream and we denote by $N$ the current number of tuples in the stream.

Given an identifier $x$, we denote the total weight of $x$ as:
$$f^{t}_{x} \triangleq \sum_{\substack{(a_i,w_i)\in\mathcal S:\\a_i=x, i\le t} }w_i.$$
When $t=N$, we denote: $f_{x}\triangleq f^{N}_{x}$.

Also, the total weight of all ids in the stream at time $t$ is:
$$R_t \triangleq \sum_{(a_i,w_i)\in\mathcal S:\\i\le t} w_i.$$

\subsection{Problem definitions:}
We now formally define the problems we address.
\begin{itemize}
	\item {\sc \textbf{$\epsilon$ - Volume Estimation}}:
	We say that an algorithm solves the \sizeEst{} problem if at any time $t$, ${Query(x)}$ returns an estimation $\widehat{f^{t}_x}$ that satisfies
	$$f^{t}_{x} \le \widehat{f^{t}_x}\le f^{t}_{x} + R_t\cdot\epsilon.$$
	
	\item {\sc \textbf{$(\theta, \epsilon)$-Elephant~Flows}}: We say that an algorithm solves the \heavyHitters{} problem if  at any time $t$, an ${Elephants()}$ query returns a set of elements $S_t$, such that for every flow $x$, $f_x^t > R_t\cdot\theta \implies x\in S_t$, and $f_x^t<R_t\cdot{\left(\theta-\epsilon\right)} \implies x\not\in S_t$. We note that if the stream is unweighted (all weights are $1$), this degenerates to the Heavy Hitters problem discussed in~\cite{SpaceSavingIsTheBest,SpaceSavings,HeavyHitters}.
\end{itemize}

\subsection{Notations}
The notations used in this paper are summarized in Table~\ref{tbl:notations}.

\newcolumntype{L}[1]{>{\raggedright\let\newline\\\arraybackslash\hspace{0pt}}m{#1}}
\newcolumntype{C}[1]{>{\centering\let\newline\\\arraybackslash\hspace{0pt}}m{#1}}
\setlength\tabcolsep{1pt}
{\renewcommand{\arraystretch}{1.8}
	
	\begin{table}[t]
		\begin{center}
			\begin{tabular}{|c|C{7cm}|}
				\hline
				Symbol & Meaning\tabularnewline
				\hline
				\hline
				$R$ & The total sum of all items in the stream so far.\tabularnewline
				\hline
				$R_t$ & The total sum of all items in the stream until time $t$.\tabularnewline
				\hline
				$f_{x}$& The weighted  frequency for item $x$ in the stream so far.\tabularnewline
				\hline
				$f^{t}_{x}$& The weighted  frequency for item $x$, at time $t$.\tabularnewline
				\hline
				$\widehat{f}^{t}_{x}$
				& An estimate of the weighted frequency of item $x$, at time $t$.\tabularnewline
				\hline
				$\gamma$ & A speed-space tradeoff parameter, affects frequency of maintenance operations and memory consumption.\tabularnewline
				\hline
				$\epsilon$ & Accuracy parameter, small $\epsilon$ means lower estimation error but more counters. \tabularnewline
				\hline
				$q_t$ & The estimate of the top $\left\lceil\frac{1}{1+\gamma}\right\rceil$ quantile at time $t$. \tabularnewline
				\hline
				$\theta$ & Threshold for elephant flows. \tabularnewline
				\hline
				$\delta$ & The error probability allowed in randomized algorithms. \tabularnewline
				\hline 			
			\end{tabular}
		\end{center}\smallskip
		\caption{Summary of notations used in this paper.}
		\label{tbl:notations}
		
	\end{table}

\section{Solution}
\subsection{Intuition}
Ideally, when the table is full, we wish to evict the smallest flow as in Space Saving.
However, in the weighted case, Space Saving runs in logarithmic time.
We achieve constant runtime by relaxing the memory constraint; instead of removing the minimal flow upon arrival of a non resident flow, we periodically remove many small flows from the table at once.
This maintenance operation takes linear time, but it is invoked infrequently enough to achieve amortized $O(1)$ runtime.
The table size only increases by a constant factor, so our solution remains (asymptotically) space optimal.

In principle, an optimal algorithm requires $O(\frac{1}{\epsilon})$ space~\cite{BerindeCormodeIndykStrauss09}, and runs in constant time. The algorithm provides an approximation guarantee that the error is bounded by $R\epsilon$ where $R$ is the total weight of all flows in the stream.

Intuitively, our periodic maintenance process identifies the $\left\lceil\frac{1}{\epsilon}\right\rceil^{th}$ largest counter, and evicts all flows whose counters is smaller than that counter.
This is obtained by internally splitting the table into two sub-tables named \emph{Active} and \emph{Passive}.
Each table entry contains both the flow's id and a counter that represents the flow's  volume.
The Active table is mutable, while the Passive is read-only.
That is, updates only affect the Active table while queries consider both of them.
The maintenance process copies all counters that are larger than the quantile of the $\left\lceil\frac{1}{\epsilon}\right\rceil^{th}$ largest flow from the Passive table to the Active table, clears the Passive table, and then switches between them.
That way, the size of the Active and Passive tables remains bounded while the $\left\lceil\frac{1}{\epsilon}\right\rceil$ largest flows are never evicted.

Each table is configured to store at most:
$T\triangleq\left\lceil\frac{\gamma}{\epsilon}\right\rceil+\left\lceil\frac{1}{\epsilon}\right\rceil-1$ items, where $\gamma$ is a positive constant that affects a speed-space trade off as explained below.
In addition, we keep a single floating point variable $q$, which represents the value of the last quantile calculated.
$q$ is the default estimation returned when querying the frequency of non-resident flows.
It ensures that such estimates meet the maximum error guarantee.

We then improve the complexity to worst case $O(1)$ with a \emph{de-amortization} process.
Intuitively, this is done by performing a small piece of the maintenance prior to each update.
We need to make sure that maintenance is always finished by the time the Active table fills up.
Once maintenance is over, we switch the Active table with the (now empty) Passive table.

\begin{figure*}[tbp]
	\begin{tabular}{ccc}
		\subfloat[]{\includegraphics[width = \matrixCellWidth]
			{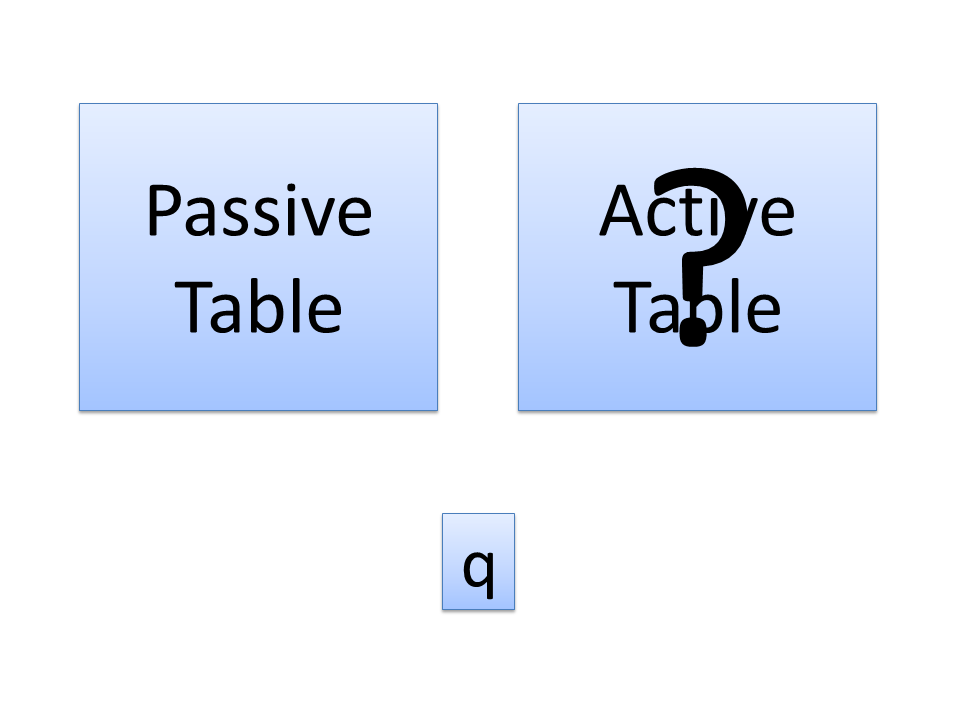}\label{fig:query1}} &
		\subfloat[]{\includegraphics[width = \matrixCellWidth]
			{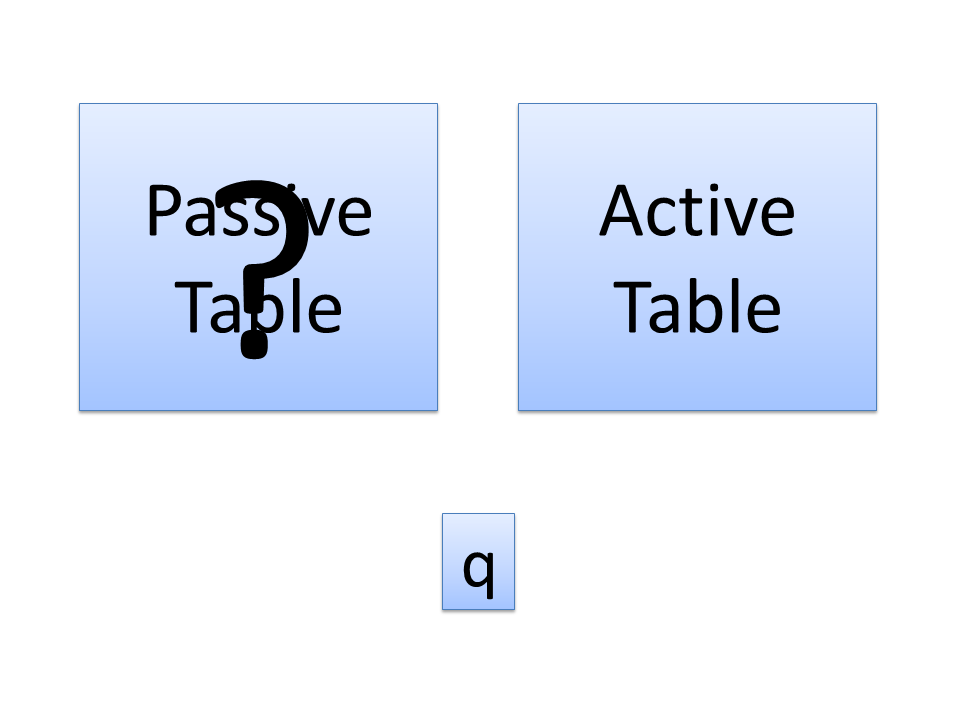}\label{fig:query2}} &
		\subfloat[]{\includegraphics[width = \matrixCellWidth]
			{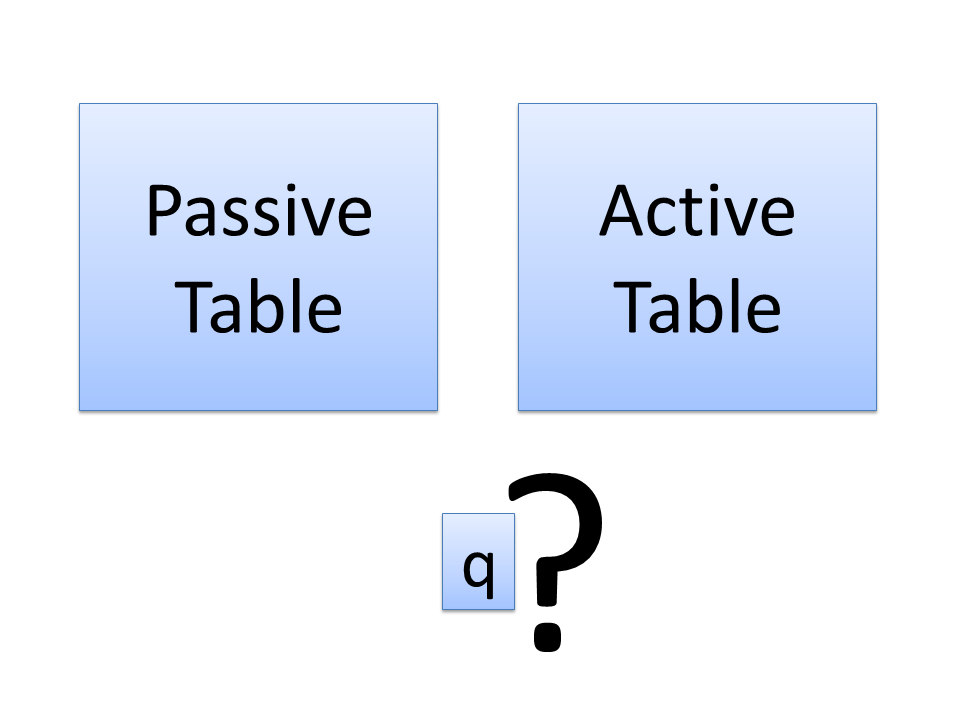}\label{fig:query3}}
	\end{tabular}
	\caption{\label{fig:query} An illustration of the query process: first we query the Active table, then if the item was not found, we query the Passive table, and finally we return $q$ if the item was not found in the Passive table as well.}
\end{figure*}

\subsection{Iterative Median Summing}
\label{algorithm}
We describe our \emph{Iterative Median SUMming (IM-SUM)} algorithm and its deamortized variant \emph{Deamortized Iterative Median SUMming (DIM-SUM)}, by detailing four operations: a query operation, which computes the volume of a flow; an elephants method that returns all elephant flows; an update operation, which increases the volume of a flow corresponding to a recently arrived packet; and a maintenance operation, which periodically discards infrequent flows to prevent overflowing.

\paragraph{$Query(x)$} The query procedure appears in Figure~\ref{fig:query}.
As depicted, the estimated weighted frequency of element $x$, $((\widehat{f}_{x})$, is $Active[x]$ if $x$ has an Active table entry,
$Passive[x]$ if $x$ has a Passive table entry and otherwise $q$.

\paragraph{$Elephants()$} Returns the elephant flows. First, we set $S_t \gets \emptyset$;
then, we traverse both the Active and Passive tables, adding to $S_t$ any flow $x$ for which $Query(x)\ge R_t\cdot\theta$.

\begin{figure}[h]
	\includegraphics[width = \columnwidth]
	{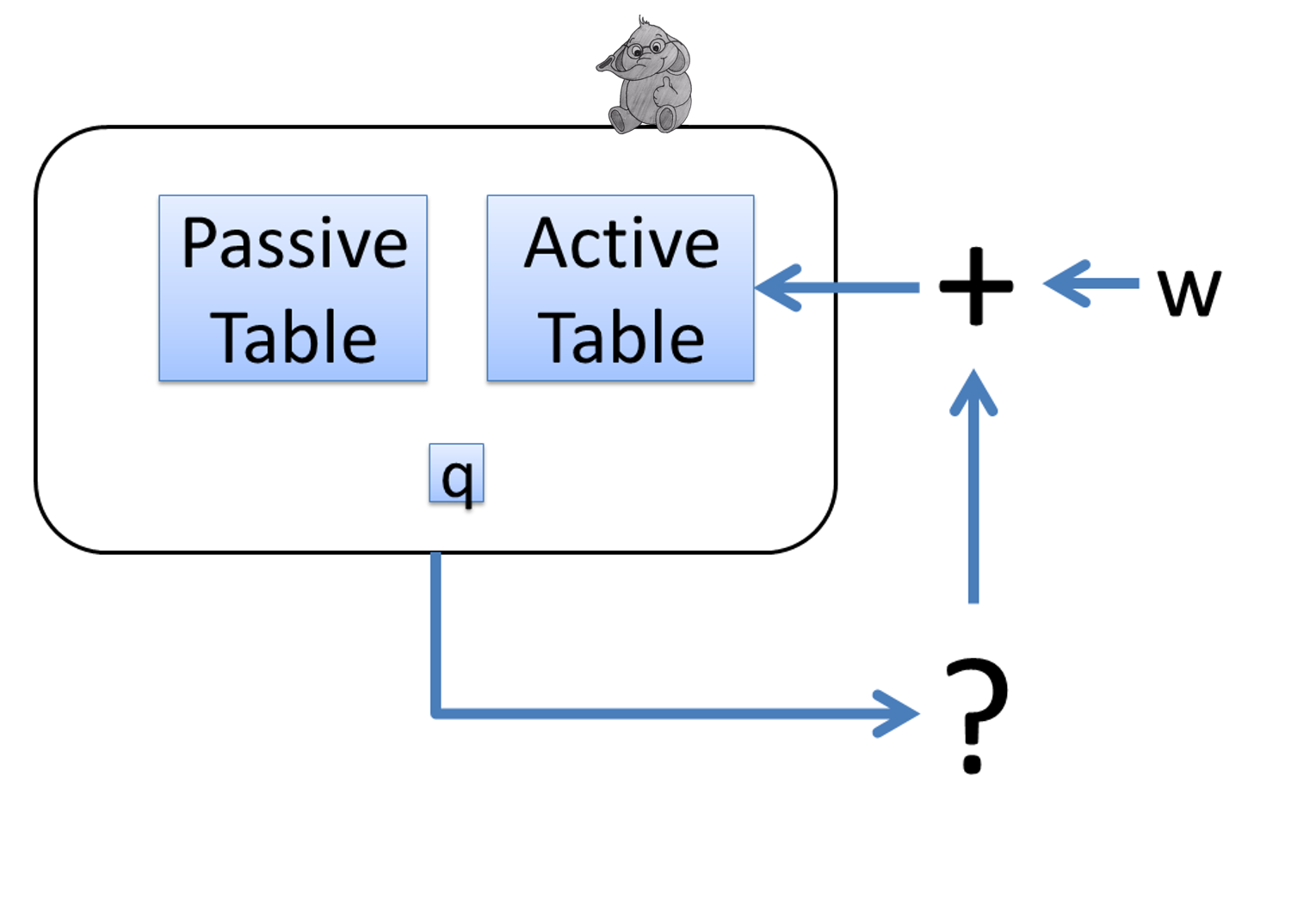}
	\caption{\label{fig:update} An illustration of the update process: we perform a query, add $w$ to the result and then update the Active table.}
\end{figure}
\paragraph{Update}
The update procedure is illustrated in Figure~\ref{fig:update}. The procedure deals with an arrival of a pair $(a_i,w)$ (flow id, weight), e.g., $a_i$ may be a TCP five tuple and $w$ the byte size of its payload.
In the update procedure, we first perform a query for $a_i$.
Let $\widehat{f}^{t-1}_{a_i}$ be the query result for $a_i$.
We add the following entry to the Active table: $\left(a_i,\widehat{f}^{t-1}_{a_i}+w\right)$.
\begin{figure*}[tbp]
	\begin{tabular}{cccc}
		\subfloat[]{\includegraphics[width = 0.25\textwidth]
			{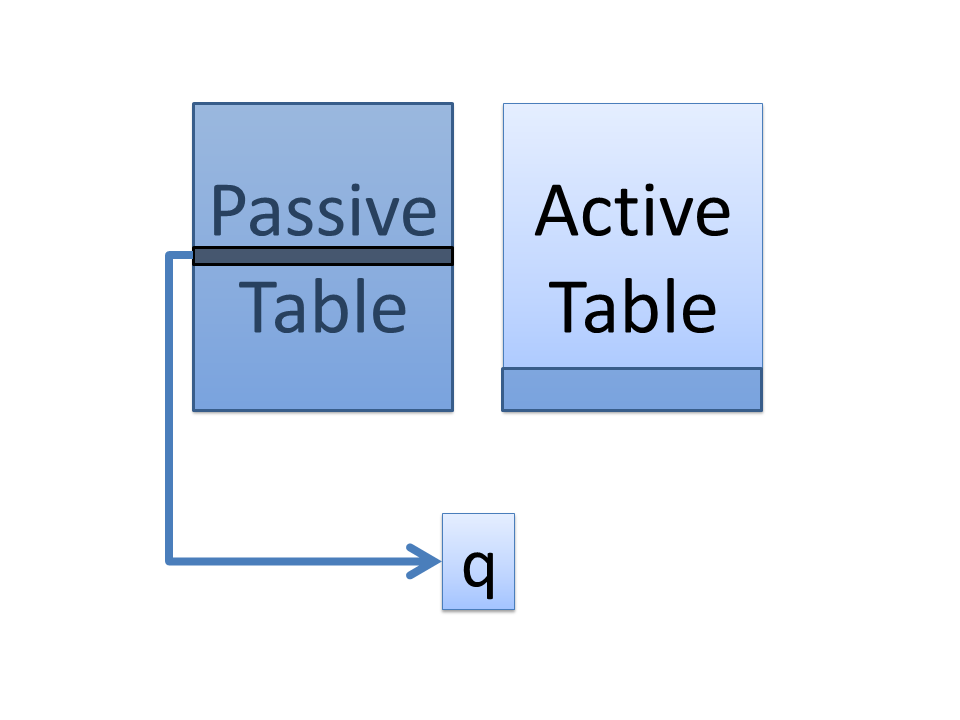}\label{fig:step1}} &
		\subfloat[]{\includegraphics[width =  0.25\textwidth]
			{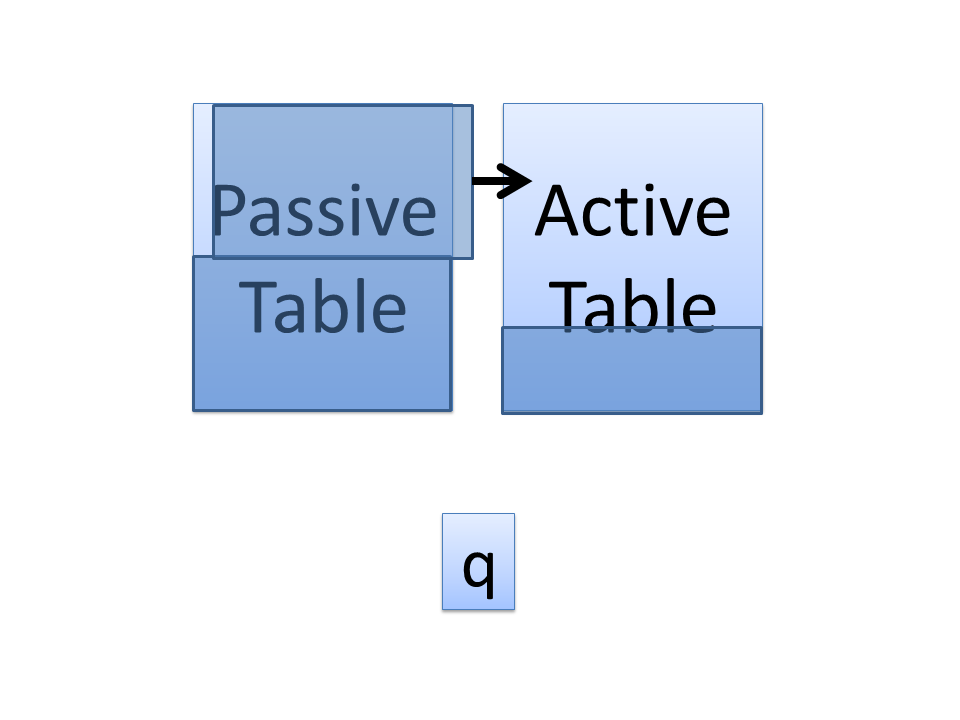}\label{fig:step2}} &
		\subfloat[]{\includegraphics[width =  0.25\textwidth]
			{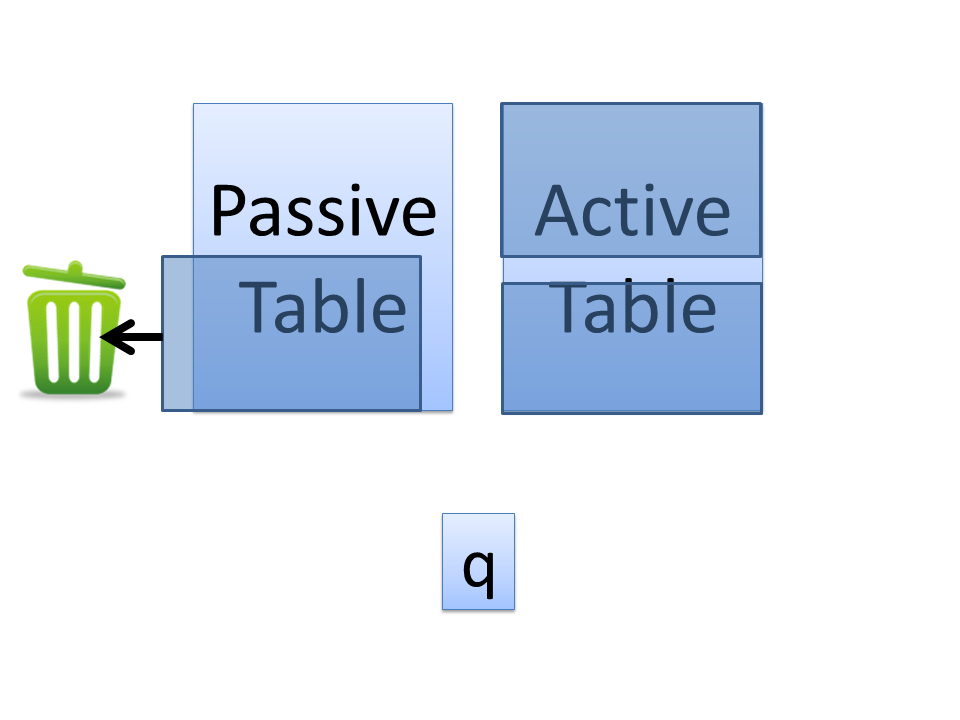}\label{fig:step3}}&
		\subfloat[]{\includegraphics[width =  0.25\textwidth]
			{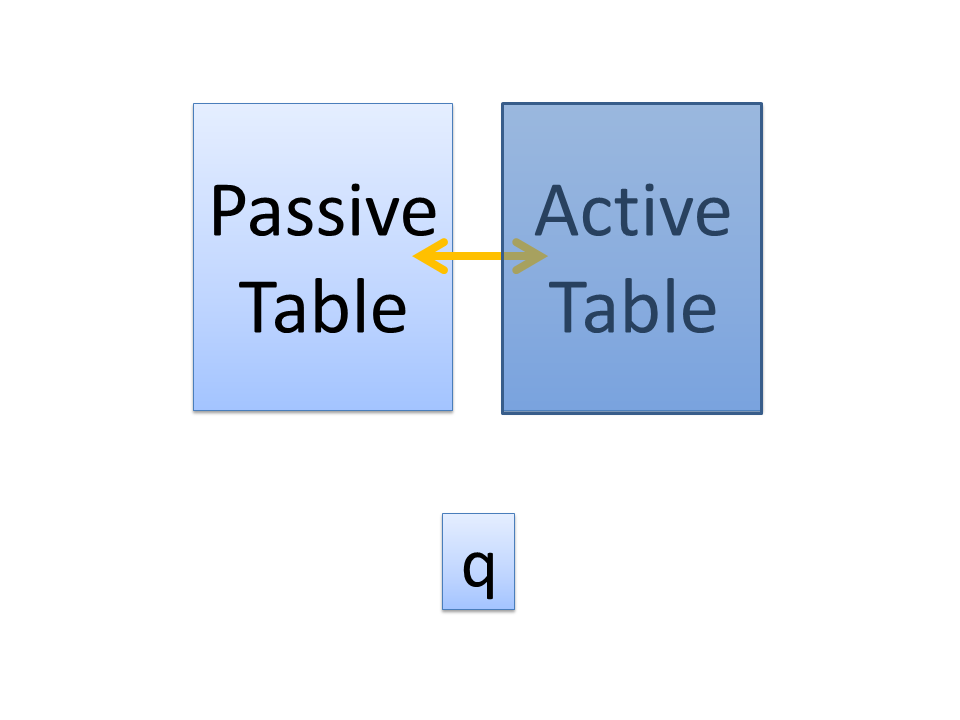}\label{fig:step 4}}
	\end{tabular}
	\caption{\label{fig:maintenance} An illustration of the maintenance process. In~\ref{fig:step1}, we find a quantile for the Passive table (a median in this example where $\gamma = 1$). Next, in~\ref{fig:step2}, we copy all entries larger than the quantile to the Active table. Then, in~\ref{fig:step3}, we clear all remaining entries from the Passive table. Finally, in~\ref{fig:step 4}, we swap the Active and the Passive tables.}
\end{figure*}

\paragraph{Maintenance}
As mentioned above, we require periodic maintenance to keep the Active table from overflowing.
We define \emph{generation} in the stream; a generation starts/ends every time the Active table fills up.
Before a generation ends,
the maintenance process needs to achieve the following:

\begin{enumerate}
	\item \label{step:quantile} Set $q$ to be the $\left\lceil\frac{1}{\epsilon}\right\rceil^{th}$ largest value in the Passive table.
	\item \label{step:move} Move or copy all items with frequency larger than $q$ from the Passive table to the Active table, except for the items that already appear in the Active table.
	\item \label{step:clear} Clear the Passive table.
	\item \label{step:replace} When the Active Table has filled, switch the Active and Passive tables before a new generation begins. After the switch, the Active table is empty and the Passive contains all entries that were a part of the Active table.
\end{enumerate}
The maintenance process is illustrated in Figure~\ref{fig:maintenance}.
Maintenance operations are performed once per generation and we need to guarantee that the number of items in the Active table never exceeds the allocated size of $T\triangleq\left\lceil\frac{\gamma}{\epsilon}\right\rceil+\left\lceil\frac{1}{\epsilon}\right\rceil-1$.
We later prove that the number of update operations between two consecutive maintenance operations is at least $F\triangleq\lceil\frac{\gamma}{\epsilon}\rceil$, and show that as a result the update time is constant if $\gamma$ is constant.

\subsection{Maintenance Implementation} \label{maintenance}
There are several ways to implement the maintenance steps described in Section~\ref{algorithm}, which provide different runtime guarantees.
Section~\ref{sec:one-piece} explains that when the maintenance operation is run serially after the Active table fills up we achieve an amortized $O(1)$ runtime, while Subsection~\ref{sec:alg-de-amortized}
explains how to improve the runtime to $O(1)$ worst case.

\subsubsection{Simple Maintenance}
\label{sec:one-piece}
In this suggestion, we perform the maintenance procedure serially at the end of each generation.
The most computationally intensive step is Step 1, where we need to find the $\lceil\frac{1}{\epsilon}\rceil^{th}$ largest counter.
When $\gamma = 1$, this value is the median, and as long as $\gamma$ is constant, it is a certain percentile. It is well known that median (and other percentiles) can be calculated in linear time in a deterministic manner, e.g., by using the median of medians algorithm. Therefore, the time complexity of Step~\ref{step:quantile} is linear with the table size of $O\left(\frac{1+\gamma}{\epsilon}\right)$.

In Step~\ref{step:move}, we copy at most $\frac{1}{\epsilon}$ entries from the Passive table to the Active table.
The complexity of this step is $O(\frac{1}{\epsilon})$.
Step 3 can be performed in $O(1)$ time or in $O\left(\frac{1+\gamma}{\epsilon}\right)$ time, depending if the hash table supports constant time flush operations.
Finally, Step 4 requires $O(1)$ time.
Therefore, since the number of update operations in a generation is $\lceil\frac{\gamma}{\epsilon}\rceil$, as shown in Lemma~\ref{frameSize}, the average update time per packet is $O(1+\frac{1}{\gamma})$, which is $O(1)$ when $\gamma$ is constant.
Therefore, the simple maintenance procedure can be executed with $O(1)$ amortized time.
As mentioned above, we nickname our algorithm with the simple maintenance procedure Iterative Median SUMming (IM-SUM).

\subsubsection{De-amortized Implementation}
\label{sec:alg-de-amortized}
In deamortization, we can perform only a small portion of the maintenance procedure on each packet arrival.
The challenge is to evenly split the maintenance task between update operations.
Our deamortized algorithm, nicknamed De-amortized Iterative Median SUMming (DIM-SUM), uses the following two techniques.

\paragraph{Controlled Execution}
\label{sec:controlled execution}
Here, we execute $O(1)$ maintenance operations after processing each packet.
Since the total maintenance time is $O\left(\frac{1+\gamma}{\epsilon}\right)$, we divide the maintenance to tasks of $O\left(1+\frac{1}{\gamma}\right)$ operations each and execute a single task for each update operation.
Therefore, when $\gamma$ is chosen to be a constant, we achieve $O(1)$ worst case complexity.

Note that we do not know how many update operations will actually occur before the Active table fills, e.g., if a flow with an Active table entry is updated, then the number of Active table entries does not increase.
This motivates us to dynamically adjust the size of the maintenance operations according to the workload, to avoid a large variance in the length of update operations.
If we use the static bound proven above, the first update operations will be slow, and once maintenance is over, update operations will be fast.
To avoid this situation, we suggest an approach to dynamically adjust the task size according to the actual workload.

To do so, prior to each update operation, we recalculate the minimum number $U$ of remaining updates according to the actual number of entries in the Active table and the maximal number of entries that still need to be copied from the Passive table.
Next, we calculate the maximum number of operations $M$ left in the current maintenance iteration.
Before each update operation, we execute $\lceil \frac{M}{U}\rceil$ maintenance operations.

\paragraph{Double Threading}
\label{sec:double threading}


A second approach is to run some of the maintenance steps in a second thread in parallel to the update operations.
We should ascertain that a sufficient number of maintenance operations are executed with every update operation.
This resembles a producer-consumer problem, where the maintenance thread is the producer and the update thread is the consumer. It can be implemented with a semaphore.
After a sufficient number of maintenance operations is executed, the semaphore increments, and before an update operation is executed, the semaphore decrements.

While the above technique requires synchronization between the two threads, it has a few advantages.
First, it simplifies the implementation.
The calculation of a quantile is usually performed with a recursive function.
Therefore, freezing the calculation and resuming it later requires saving the trace of recursive calls.
Second, the double-threading procedure enables running the update and the maintenance procedures in parallel, which can save time.
The drawback of this technique is the use of semaphore for synchronization, which hurts performance.

In our deamortized implementation, we used a combination of the two techniques.
We used controlled execution for steps~\ref{step:move} and~\ref{step:clear} and double-threading for Step~\ref{step:quantile}. Step~\ref{step:replace} is executed only at the end of a generation.

\ifdefined \EXTENDED
\subsection{What is in a Hash Table?}
The only data structure used by our solution is a hash table.
\fi
\ifdefined \NINEPAGES
Notice that the only data structure used by our solution is a hash table.
\fi
This means that the hash table's properties directly impact the performance of our solution.
In particular, advances in hash tables (e.g.,~\cite{TinyTable})  can directly improve our work.

\ifdefined \EXTENDED
Below, we survey some interesting~possibilities.

\subsubsection{Do we really need two different tables?}
In principle, one could implement both Active and Passive tables as a single table with a bit that decides if the entry belongs to the Passive or Active table.
This approach can speed up the maintenance process, as moving an item from the Passive to the Active table can be done by flipping a bit, which can even be done in an atomic manner.

\subsubsection{Cuckoo Hash Table}
Cuckoo hashing is very attractive for read intensive workloads as it offers $O(1)$ worst case queries.
The trade off is that insertions end in $O(1)$ time only with high probability.

Another interesting benefit of the Cuckoo hash table is fast delete operations; all entries below a certain threshold can be instantly removed by treating them as empty~\cite{SlidingWindowBF}.
A similar technique could also be used in our case and further reduce the time required to clear the Passive table.

\subsubsection{Compact Hash Tables}
When space is very tight, we can leverage compact hash tables.
These data structures avoid the use of pointers and instead employ bit indexes to encode the table structure~\cite{TinyTable,RankedIndexHashing}.
Clearly, leveraging such structures is expected to significantly reduce space consumption.

\subsubsection{Fast Hash Tables}
Perhaps the most obvious benefit of relying only on standard hash tables is speed. Since we do not maintain any proprietary data structures, we can potentially improve our runtime as faster hash tables are developed.
For example, \emph{Hop Scotch Hashing}~\cite{Herlihy08hopscotchhashing} is considered one of the fastest parallel hash tables in practice.
\fi
\section{Analysis}
\subsection{Correctness}
Our goal is to prove that IM-SUM and DIM-SUM solve the \sizeEst{} and \heavyHitters{} problems. To do so, we first bound the value of $q_t$.

\begin{lemma}
	$q_t\le R_t\cdot \epsilon$ at all times $t$.
	\label{lem:quantile-size}
\end{lemma}
\ifdefined \NINEPAGES
Due to lack of space, the proof of this lemma appears in the full version of this paper~\cite{dimsum-full}.
\fi
\ifdefined \EXTENDED
\begin{proof}[Proof of Lemma~\ref{lem:quantile-size}]
	Let $S_t$ be the sum of the $\left\lceil\frac{1}{\epsilon}\right\rceil$ items in the tables with the largest estimation value at time $t$. We show by induction on $i$ that $S_i\le R_i$ for any time $i$. Assume that $S_{i-1}\le R_{i-1}$.
	
	When an item $x$ arrives at time $i$, if it is already one of the $\left\lceil\frac{1}{\epsilon}\right\rceil$ largest items, then both $R_{i-1}$ and $f^{t}_{x}$ increase by $w_i$, and consequently $S_i=S_{i-1}+w_i\le R_{i-1}+w_i=R_i$.
	
	If $x$ was not one of the largest items at time $i-1$ and has not become one at time $i$, then $S_i=S_{i-1}\le R_{i-1}\le R_i$.
	
	If, however, $x$ was not one of the largest items at time $i-1$ but became one at time $i$, mark the frequency of the $\left\lceil\frac{1}{\epsilon}\right\rceil^{th}$ largest item by $Q_{i-1}$.
	If there are fewer than $\left\lceil\frac{1}{\epsilon}\right\rceil$ items in the tables, define $Q_{i-1}=q_{i-1}$.
	After adding $w_i$, the estimation of $x$ ($\widehat{f}^{t}_x$) becomes at most $Q_{i-1}+w_i$.
	Therefore $\widehat{f}^{t}_x$ adds to $S_{i-1}$ up to $Q_{i-1}+w_i$.
	Fortunately, if $Q_{i-1}>0$, then there must be at least $1\over\epsilon$ items with a positive estimation, and since $x$ has become one of the largest items, it must have replaced an item with frequency at least $Q_{i-1}$.
	Therefore, we subtract at least $Q_{i-1}$ from $S_{i-1}$. In case $Q_{i-1}$ was $0$, we still subtract $Q_{i-1}=0$ from $S_{i-1}$.
	Either way, $$S_i\le S_{i-1}+Q_{i-1}+w_i-Q_{i-1}=S_{i-1}+w_i\le R_{i-1}+w_i=R_i.$$
	Thus $S_i\le R_i$ for every $i$.
	
	We next consider the time $g$ at the beginning of a generation before $q_t$ was computed. There cannot be $\left\lceil\frac{1}{\epsilon}\right\rceil$ items larger than ${R_g\over \left\lceil\frac{1}{\epsilon}\right\rceil}\le R_g\epsilon$, or otherwise $S_g$ would be larger than $R_g$. Since this is the beginning of the generation, all of the items are in the passive table, and the quantile will be computed on the current items. Therefore $q_t$ will be smaller or equal to $R_g\epsilon$. Since the sum of elements only grows, $q_t\le R_g\epsilon\le R_t\epsilon$.
\end{proof}
\fi

\begin{theorem}\label{thm:correctness}
	Our algorithm solves \sizeEst:
	That is, at any step ($t$): ${f}^{t}_{x} \le \widehat{f}^{t}_{x} \le{f}^{t}_{x}+R_t\epsilon$.
\end{theorem}

\begin{proof}
	We prove the theorem by induction over the number of elements seen $t$.\\
	\textbf{Basis:} At time $t=0$, the claim holds trivially.\\
	\textbf{Hypothesis:} Suppose that at time $t-1$ some element $x$ has an approximation
	$$f^{t-1}_x\le\widehat{f^{t-1}_x}\le f^{t-1}_x+R_{t-1}\epsilon.$$		
	\textbf{Step:} Assume that at time $t$ item $a_t$ arrives with weight $w_t$. If $a_t = x$, then the new estimation grows by exactly $w_t$, same as the true weight of $x$. According to the induction hypothesis, the estimation remains correct. Similarly, when another item arrives at time $t$, $f^{t-1}_x$ does not change, and neither does $\widehat{f^{t-1}_x}$.
	Steps~\ref{step:quantile} ,~\ref{step:move} and ~\ref{step:replace} do not change the estimation for any item. 

	During step~\ref{step:clear}, however, $x$ could be removed from the Passive table.
	If $x$ is also in the Active table, $\widehat{f}^{t}_x$ remains the same.
	Otherwise, $\widehat{f}^{t}_x$ is changed from $\widehat{f^{t-1}_x}$ to $q$.
	However, since $x$ only appears in the Passive table, we deduce that it was not moved in step~\ref{step:move}.
	This means that the estimation for $x$ at the beginning of the generation was smaller or equal to $q$.
	Therefore, it is immediate from the induction hypothesis that $f^{t-1}_x\le q$.
	After $x$ is removed from the Passive table as well, the estimation becomes $\widehat{f^t_x}=q$.
	Hence, $f^{t}_x=f^{t-1}_x\le q=\widehat{f^t_x}$.
	
	Lemma~\ref{lem:quantile-size} proves that $q\le f^t_x+R_t\epsilon$ and therefore we have
	$f^{t}_x\le\widehat{f^{t}_x}\le f^{t}_x+R_t\epsilon$.
	
\end{proof}

Next, we prove that our algorithm is also able to identify the elephant flows.
We note that the runtime is optimal as there could be $\Omega(\frac{1}{\epsilon})$ elephant flows.

\begin{theorem}
	At any time $t$, the $Elephants()$ query solves \heavyHitters{} in $O(\frac{1}{\epsilon})$ time.
\end{theorem}

\begin{proof}
	Queries are done by traversing the tables and observing all elements $x$ such that $\widehat{f^{t}_x}>q$.
	The tables are of size $O(\frac{1}{\epsilon})$ when $\gamma$ is constant and therefore the complexity is $O(\frac{1}{\epsilon})$.
	
	Let $x$ be such that $f^t_x>R_t\cdot \theta$.
	According to Theorem~\ref{thm:correctness} and Lemma~\ref{lem:quantile-size}, $\widehat{f^{t}_x}\ge f^t_x>R_t\cdot \theta\ge q$.
	Hence, $x$ must appear in one of the tables with frequency greater than $R_t\cdot \theta$.
	Otherwise, it would have an estimation of less than $q$.
	Consequently, $x\in S_t$.
	
	Next, let $x$ be such that $f^t_x<R_t\cdot (\theta-\epsilon)$.
	According to Theorem~\ref{thm:correctness}, $\widehat{f^{t}_x}\le f^t_x+R_t\cdot\epsilon<R_t\cdot (\theta-\epsilon)+R_t\cdot\epsilon=R_t\cdot\theta$.
	Therefore, $x$ will not be included in $S_t$.
\end{proof}

\subsection{Runtime}
We now prove that IM-SUM runs at $O(1)$ amortized time and that DIM-SUM runs at $O(1)$ worst case time.
Our goal is to evaluate both the complexity of the maintenance and the number of update operations between two consecutive maintenance operations.
We start with evaluating the minimal number of update operations between two consecutive generations, i.e., the number of unique flows that must be encountered before we switch to a new generation.

\ifdefined \NINEPAGES
For lack of space, the proofs of the following two technical lemmas are deferred to the full version of this paper~\cite{dimsum-full}.
\fi
\begin{lemma}
	\label{frameSize}
	Denote by $G_{min}$ the minimal number of update operations between subsequent generations: $G_{min} = \left\lceil\frac{\gamma}{\epsilon}\right\rceil$.
\end{lemma}
\ifdefined \EXTENDED
\begin{proof}
	In the beginning of every generation the Active table is empty. During step~\ref{step:move}, the Active table receives up to $\left\lceil\frac{1}{\epsilon}\right\rceil-1$ items from the Passive table. Other than that, the Active table only grows due to update operations. Therefore, there must be at least
	\begin{multline*}
		G_{min}=T-\left(\left\lceil\frac{1}{\epsilon}\right\rceil-1\right)\\
		=\left\lceil\frac{\gamma}{\epsilon}\right\rceil+\left\lceil\frac{1}{\epsilon}\right\rceil-1-\left(\left\lceil\frac{1}{\epsilon}\right\rceil-1\right)=\left\lceil\frac{\gamma}{\epsilon}\right\rceil
	\end{multline*}
	updates in each generation.
\end{proof}
\fi

Next, we bound the number of hash table operations required to perform the maintenance process.
\begin{lemma}\label{lemma:maintenance}
	The maintenance process requires:
	$M=O\left(\frac{1+\gamma}{\epsilon}\right)$ hash table operations.
\end{lemma}
\ifdefined \EXTENDED
\begin{proof}
	We calculate the complexity of all steps. For Step~\ref{step:quantile}, it is well known that a median can be found in linear time, e.g., using the median of medians algorithm.
	The same is true for any percentile, and therefore Step~\ref{step:quantile} requires $$O\left(\frac{1+\gamma}{\epsilon}\right)$$ hash table operations to complete.
	
	Step~\ref{step:move} goes over the Passive table, which is of size $O\left(\frac{1+\gamma}{\epsilon}\right)$, and copies or moves some of the items to the Active table. Therefore, its complexity is
	$O\left(\frac{1+\gamma}{\epsilon}\right)$.
	
	Step~\ref{step:clear} can be performed in $O(1)$ and so does Step~\ref{step:replace}.
	Therefore, the total complexity of all steps is:  $$O\left(\frac{1+\gamma}{\epsilon}\right).$$
\end{proof}
\fi

\begin{theorem}
	For any fixed $\gamma$, IM-SUM runs in $O(1)$ amortized and DIM-SUM in $O(1)$ worst case hash table~operations.
\end{theorem}

\begin{proof}
	First, we observe that for any constant $\gamma$, Lemma~\ref{frameSize} establishes that there are $O(\frac{1}{\epsilon})$ update operations between generations.
	Similarly, by Lemma~\ref{lemma:maintenance}, there are $O(\frac{1}{\epsilon})$ hash function operations to do during each maintenance process.
	
	In IM-SUM, maintenance is performed once per generation and thus the amortized complexity is $O(1)$.
	
	In DIM-SUM, as the work is split between update operations, each update has to perform $O(1)$ hash table operations.
	Let $U_t$ be the minimum number of update operations left in the generation at time $t$ and let $M_t$ be the maximum number of maintenance operations left.
	We show by induction that $\frac{M_t}{U_t}=O(1)$.
	At a $t$ in the beginning of the generation, $U_t=G_{min}$ and $M_t=M$.
	Therefore, $\frac{M_t}{U_t}=O(1)$.
	
	Assume, by induction, that for any time $t$, $\frac{M_t}{U_t}=O(1)$.
	After a single update, $U_{t+1}\ge U_t-1$, because otherwise the minimum number of updates at time $t$ would be smaller than $U_t$.
	$M_t$ decreases to at most $M_{t+1}\le M_t-\frac{M_t}{U_t}=M_t\left(1-\frac{1}{U_t}\right)$.
	Therefore, the number of maintenance operations executed at time $t+1$ is $$O\left(\frac{M_{t+1}}{U_{t+1}}\right)\le O\left(\frac{M_t\left(1-\frac{1}{U_t}\right)}{U_t-1}\right)=O\left(\frac{M_t}{U_t}\right)=O(1).$$
	
	The query process requires 2 hash table operations and is therefore also $O(1)$.
\end{proof}

\subsection{Required Space}
We now show that for each constant $\gamma$, IM-SUM and DIM-SUM require the (asymptotically) optimal $O(\frac{1}{\epsilon})$ table entries.

\begin{theorem}
	For each constant $\gamma$, the space complexity of IM-SUM and DIM-SUM is $O(\frac{1}{\varepsilon})$.
\end{theorem}

\begin{proof}
	The proof follows from the selected table sizes.
	Each table is sized to contain a maximum of $\frac{\gamma}{\epsilon} + \frac{1}{\epsilon}$ entries and therefore the total number of table entries is $O(\frac{1}{\epsilon})$.
\end{proof} 
\section{Evaluation}
\begin{figure*}[t]
	\begin{tabular}{ccc}
		\subfloat[Chicago]{\includegraphics[scale=0.5] 
			{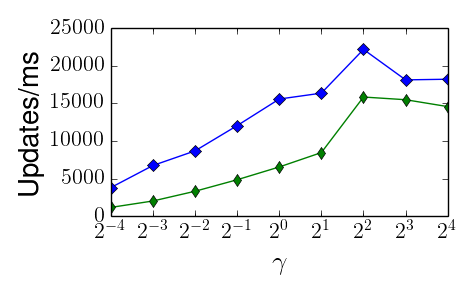}} &	
		\subfloat[SanJose]{\includegraphics[scale=0.5] 
			{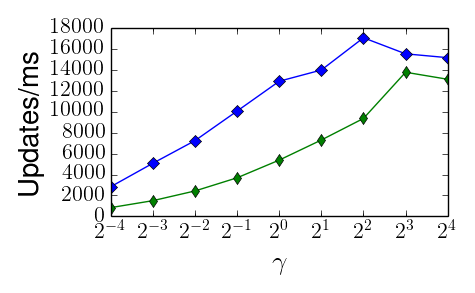}} &	
		\subfloat[YouTube]{\includegraphics[scale=0.5] 
			{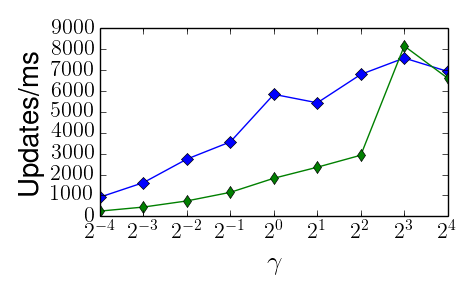}} \\	
		\subfloat[UCLA-TCP]{\includegraphics[scale=0.5] 
			{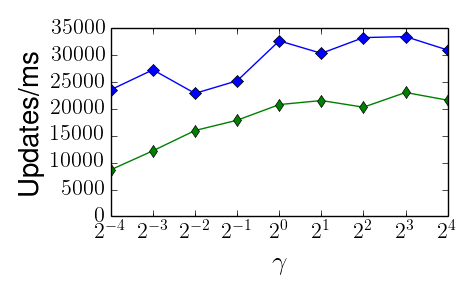}} &
		\subfloat[Legend]{\includegraphics[scale=0.5] 
			{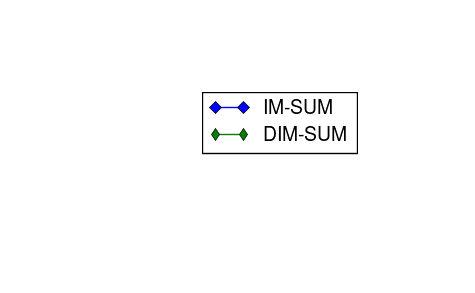} \label{fig:legend2}} 				& 		
		\subfloat[U1CLA-UDP]{\includegraphics[scale=0.5] 
			{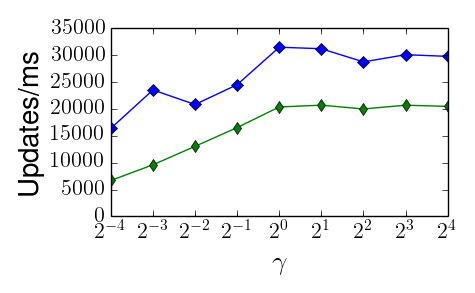}}\\
		\subfloat[Zipf0.7]{\includegraphics[scale=0.5] 
			{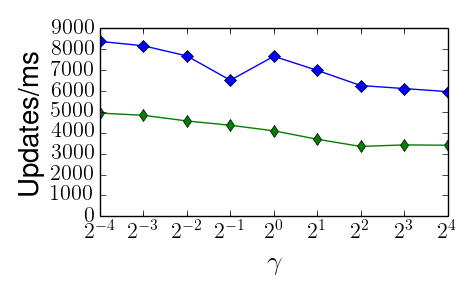}} 	
		&
		\subfloat[Zipf1]{\includegraphics[scale=0.5] 
			{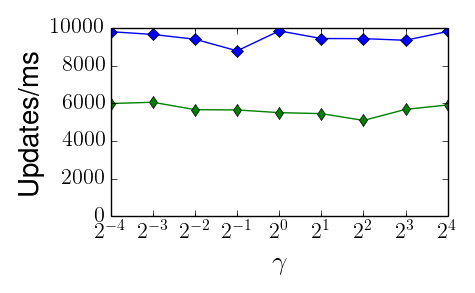}}
		&
		\subfloat[Zipf1.3]{\includegraphics[scale=0.5] 
			{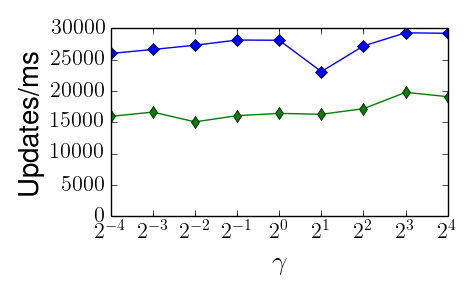}}

	\end{tabular}
	
	\caption{	\label{fig:gamma}The effect of the parameter $\gamma$ on the performance of the algorithms over different traces ($\epsilon = 2^{-15}$).}
\end{figure*}

\begin{figure*}[tp!]
	\begin{tabular}{ccc}
		\subfloat[Chicago]{\includegraphics[scale=0.5] 
			{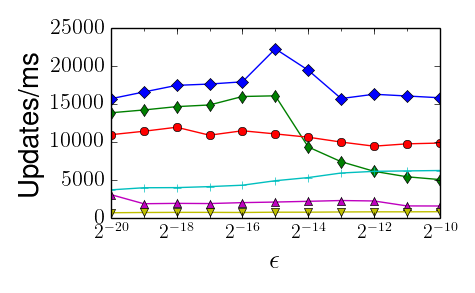}} &
		\subfloat[SanJose]{\includegraphics[scale=0.5] 
			{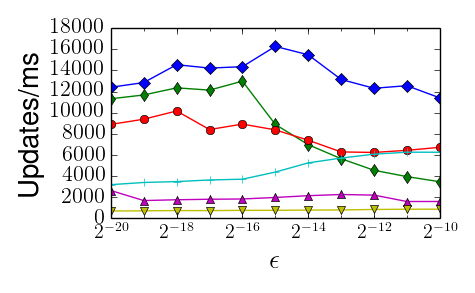}} &	
		\subfloat[YouTube]{\includegraphics[scale=0.5] 
			{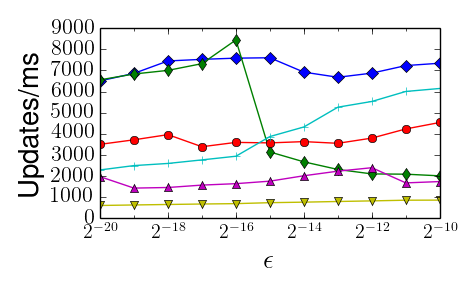}}\\	
		\subfloat[UCLA-TCP]{\includegraphics[scale=0.5] 
			{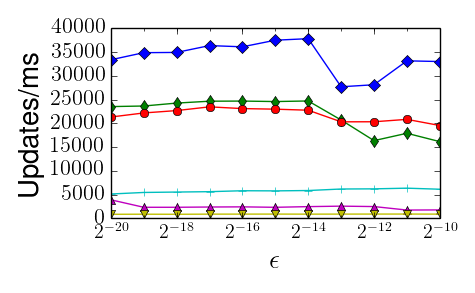}} &
		\subfloat[legend]{\includegraphics[scale=0.5] 
			{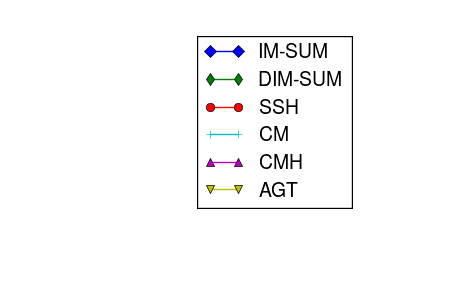}}&
		\subfloat[UCLA-UDP]{\includegraphics[scale=0.5] 
			{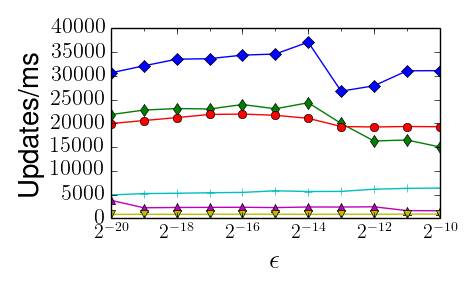}}

	\end{tabular}
	\caption{\label{fig:epsilon} Comparison of the number of updates/ms performed by the different algorithms for a given error bound.}
\end{figure*}

We now present an extensive evaluation for comparing our proposed algorithm with the leading alternatives.
We only considered alternative algorithms that are able to count packets of variable sizes and are able to answer queries on-line.

\subsection{Datasets}
Our evaluation includes the following datasets:
\begin{enumerate}
	\item The CAIDA Chicago Anonymized Internet Trace 2015~\cite{CAIDA} , denoted \emph{Chicago}. The trace was collected by the `equinix-chicago' high-speed monitor and contains a mix of 6.3M TCP, UDP and ICMP packets. The weight of each packed is defined as the size of its payload, not including the header.
	\item The CAIDA San Jose Anonymized Internet Trace 2014~\cite{CAIDASJ} , denoted \emph{SanJose}. The trace was collected by the `equinix-sanjose' high-speed monitor and contains a mix of 20.2M TCP, UDP and ICMP packets.
	\ifdefined \EXTENDED
	The weight of each packet is defined as the size of its payload, not including the header.	
	\fi
	\ifdefined \NINEPAGES
	The weight of each packet is defined as in \emph{Chicago}.
	\fi
	\item The UCLA Computer Science department TCP packet trace (denoted \emph{UCLA-TCP})~\cite{UCLA}. This trace contains 16.5M TCP packets passed through the border router of the Computer Science Department, University of California, Los Angeles.
	\ifdefined \EXTENDED
	The weight of each packet is defined as the size of its payload, not including the header.	
	\fi
	\ifdefined \NINEPAGES
	Weights of packets is defined as in \emph{Chicago}.
	\fi
	\item \emph{UCLA-UDP}~\cite{UCLA}. This trace contains 18.5M UDP packets passed through the border router of the Computer Science Department, UCLA.
	\ifdefined \EXTENDED
	The weight of each packet is defined as the size of its payload, not including the header.	
	\fi
	\ifdefined \NINEPAGES
	Weights are defined as in \emph{Chicago}.
	\fi
	\item YouTube Trace (referred to as \emph{YouTube})~\cite{youtube}. The trace contains a sequence of 436K accesses to YouTube's videos made from within the University of Massachusetts Amherst. The weight of each video is defined as its length (in seconds).
	\ifdefined \EXTENDED
	An example application for such measurement is the caching of videos according to their bandwidth usage to reduce traffic.
	\fi
	\ifdefined \NINEPAGES
	An example application for this is caching videos according to their bandwidth usage.
	\fi
	\item (Unweighted) Zipf streams that contain a series of i.i.d elements sampled according to a Zipfian distribution of a given skew. We denote a Zipf stream with skew $X$ as \emph{Zipf$X$}.
	Each packet is of weight $1$.
\end{enumerate}

\begin{figure*}[tbp]
	\centering
	\begin{tabular}{ccc}
		\subfloat[$\epsilon=2^{-10}$]{\includegraphics[scale=0.5 ]
			{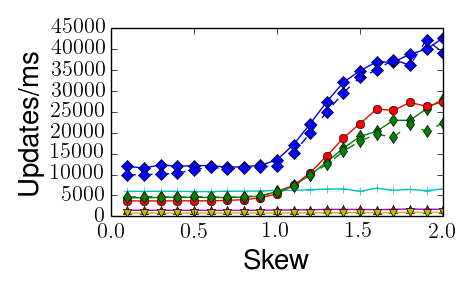}} &
		\subfloat[Legend]{\includegraphics[scale=0.5 ]
			{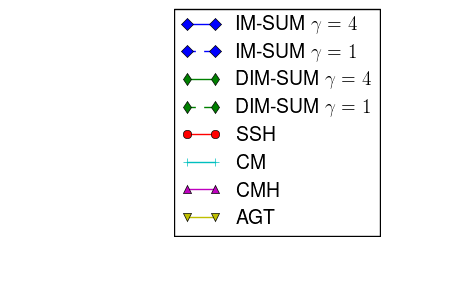}}&
		\subfloat[$\epsilon=2^{-20}$]{\includegraphics[scale=0.5] 
			{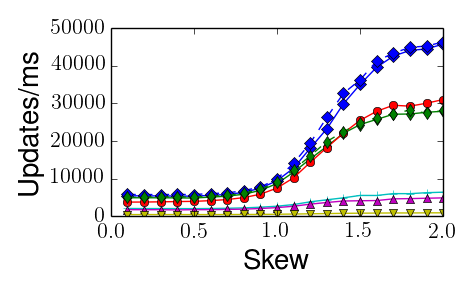}}
	\end{tabular}
	\caption{\label{fig:skew} Comparison of the number of updates/ms performed by the different algorithms for a given skew.}
\end{figure*}

\subsection{Implementation}
We compare both IM-SUM and DIM-SUM~\cite{OpenSourceAnonimous} to Count Min Sketch (CM), a Hierarchical CM-sketch (CMH), which is an extended version of CM suited to find frequent items~\cite{CMH}, a robust implementation of the Count Sketch that has been extended using Adaptive Group Testing by~\cite{SpaceSavingIsTheBest} to support finding Frequent Items (AGT), and a heap based implementation of Space Saving (denoted SSH)~\cite{SpaceSavingIsTheBest2010,SpaceSavingIsTheBest2009}.
The code for the competing algorithms is the one released by~\cite{SpaceSavingIsTheBest}.
All implementations are in C++ and the measurements were made on a 64-bit laptop with a 2.30GHz CPU, 4.00 GB RAM, 2 cores and 4 logical cores.
In our own algorithms, we used the same hash table that was used to implement SSH for fairness.

\subsection{The effect of $\gamma$}

Recall that $\gamma$ is a performance parameter; increasing $\gamma$ increases our space requirement but also makes maintenance operations infrequent and therefore potentially improves performance.
Figure~\ref{fig:gamma} shows the performance of IM-SUM and DIM-SUM for values of $\gamma$ ranging from $2^{-4}$ to $2^{4}$ and $\epsilon=2^{-15}$.
The YouTube trace seems to benefit from larger $\gamma$ until $\gamma =8$, Chicago until $\gamma =4$, and UCLA seems to saturate around $\gamma =1$.
SanJose peaks for IM-SUM at $\gamma=4$ and $\gamma=8$ for DIM-SUM.
Zipf1.3 increases with $\gamma$ until $\gamma=8$.
The Zipf0.7 and Zipf1 traces, however, slightly decrease with $\gamma$.
We explain the decline by an increase in memory consumption, which slows down memory access.
The decline is probably not evident in other traces, because in high skew traces the effect of reducing the frequency of the maintenance procedure is larger than the effect of decreasing the memory requirement.
In order to balance speed and space efficiency, we continue our evaluation with $\gamma=4$.
We also recommend this setting for traces with unknown characteristics.
If additional memory is available, reducing $\epsilon$ is more effective than increasing $\gamma$.

It is also evident that our de-amortization process has significant overheads.
We attribute this to the synchronization overhead between the two threads of DIM-SUM.

\subsection{Effect of $\epsilon$}

Since memory consumption is coupled to accuracy~$\epsilon$, it is interesting to evaluate the operation speed as a function of $\epsilon$.
As IM-SUM and DIM-SUM are asymptotically faster, we expect them to perform better than known approaches for small $\epsilon$s (that imply a large number of counters).
Figure~\ref{fig:epsilon} depicts the operation speed of the algorithms,
as a function of $\epsilon$.
As shown, IM-SUM is considerably faster than all alternatives in all tested workloads, even for large values of $\epsilon$.
The results for DIM-SUM are mixed.
For large $\epsilon$ values, it is mostly slower than the alternatives.
However, for small $\epsilon$s, it is consistently better than previous works.

\subsection{Effect of Skew on operation speed}

Intuitively, the higher the skew of the workload, the less frequent maintenance operations occur, because it takes longer before we see enough unique flows for the Active table to fill up.
Figure~\ref{fig:skew} demonstrates the effect of the zipf skew parameter on performance.
As can be observed, while CMH and AGT are quite indifferent to skew, IM-SUM and DIM-SUM benefit greatly from it and IM-SUM is faster than all alternatives.
DIM-SUM is consistently faster than CMH and AGT, and has similar speed to SSH. CM is slightly faster than DIM-SUM for low skew workloads but more than three times slower than DIM-SUM when the skew is high.


\section{Discussion}
In this paper, we have shown two variants of the first (asymptotically) space optimal algorithm that can estimate the total traffic volume of every flow as well as identify the elephant flows (in terms of their total byte count).
The first variant, IM-SUM, is faster on the average case but only ensures $O(1)$ amortized execution time, while the second variant, DIM-SUM, offers $O(1)$ worst case time guarantee.
We have benchmarked our algorithms on both synthetic and real-world traces and have demonstrated their superior performance.

Looking into the future, we hope to further reduce the actual running time of DIM-SUM to match the rate of IM-SUM.
This would probably require redesigning our code to eliminate synchronization, possibly by using an existing concurrent hash table implementation.
A C++ based open source implementation of this work and all other code used in this paper is available at~\cite{OpenSourceAnonimous}. 

\bibliographystyle{plain}
\bibliography{references}

\end{document}